\documentclass[conference,letterpaper]{IEEEtran}
\usepackage{mathtools,amsmath,amssymb,amsthm}
\usepackage{tikz,pgfplots,myTikzStyle}
\usetikzlibrary{patterns}
\usepgfplotslibrary{ternary}
\usepackage{esdiff}
\usepackage{graphicx}
\usepackage{comment}

\usepackage[utf8]{inputenc} 
\usepackage[T1]{fontenc}
\usepackage{ifthen}
\usepackage{cite}
\DeclareMathOperator*{\argmin}{arg\,min}
\DeclareMathOperator{\supp}{supp}

\newcommand{\set}[1]{\mathcal{#1}}
\newcommand{\rv}[1]{\mathsf{#1}}

\newcommand{\done}[2]{d_1\left( #1 , #2 \right)}
\newcommand{\E}[2]{\mathbb{E}_{#1}\left[#2\right]}

\newcommand{\Var}[2]{V_{#1}\left(#2\right)}
\newcommand{\entrp}[1]{H\left(#1\right)}
\newcommand{\crossentrp}[2]{X\left( #1 \Vert #2 \right)}
\newcommand{\diverg}[2]{D\left( #1 \Vert #2 \right)}

\newcommand{\Rinfo}{R_{\text{info}}}

\newcommand{\pmf}[1]{P_{\rv{#1}}}
\newcommand{\qmf}[1]{Q_{\rv{#1}}}

\newtheorem{theorem}{Theorem}
\newtheorem{lemma}{Lemma}

\newtheorem{corollary}{Corollary}

\newlength{\dotsize}
\newlength{\dotspread}
\tikzset{dotsize/.code={\setlength{\dotsize}{#1}},
         dotspread/.code={\setlength{\dotspread}{#1}}}
\tikzset{dotsize=.5pt,
         dotspread=3pt}
\makeatletter
\pgfdeclarepatternformonly[\dotspread,\dotsize]
  {custom dots}
  {\pgfqpoint{-2\dotsize}{-2\dotsize}}
  {\pgfqpoint{2\dotsize}{2\dotsize}}
  {\pgfqpoint{\dotspread}{\dotspread}}
  {
  \pgfsetcolor{\tikz@pattern@color}
  \pgfpathcircle{\pgfqpoint{0pt}{0pt}}{\dotsize}
  \pgfusepath{fill}
  }
 \makeatother

\theoremstyle{definition}
\newtheorem{example}{Example}

\begin{document}

\title{Divergence Scaling for Distribution Matching}

\author{%
  \IEEEauthorblockN{Gerhard~Kramer}
  \IEEEauthorblockA{Institute for Communications Engineering\\
                    Technical University of Munich\\
                    D-80333 Munich, Germany}
}

\maketitle

\begin{abstract}
Distribution matchers for finite alphabets are shown to have informational divergences that grow logarithmically with the block length,
generalizing a basic result for binary strings. 
\end{abstract}

\section{Introduction}
\label{sec:introduction}
Distribution matchers (DMs) are random number generators with a one-to-one mapping~\cite{schulte2020invertible}.
An important application of DMs is probabilistic amplitude shaping for energy-efficient communication~\cite{bocherer2015bandwidth}.
The objective of this paper is to prove that block (fixed-length) DMs cannot give low informational divergence
for finite alphabets, thereby generalizing a basic result of Schulte and Geiger~\cite{schulte2017divergence} for binary strings.
The main tool we use is a large deviation theorem of Bahadur and Ranga Rao~\cite{bahadur1960-2} that applies more generally than the binomial bounds and identities in~\cite{schulte2017divergence}.

This paper is organized as follows. 
Sec.~\ref{sec:preliminaries} introduces notation and a few lemmas.
Sec.~\ref{sec:bounds-sums} gives bounds on sums and establishes two further results, namely Lemma \ref{theorem:brr-size} that bounds the size of optimal DM codebooks and Theorem \ref{theorem:scaling-bounds} that bounds a divergence for average empirical distributions.
Sec.~\ref{sec:DM-scaling} develops the desired scaling results.
Appendixes~\ref{app:entropy-inequality2}-\ref{app:scaling-bounds} provide proofs of Lemmas and a Theorem.

\section{Preliminiaries}
\label{sec:preliminaries}
\subsection{Basic Notation}

Sets are written with calligraphic letters $\set{A}$ and the empty set with $\emptyset$.
The cardinality of $\set{A}$ is $|\set{A}|$ and the $n$-fold Cartesian product of $\set{A}$ is $\set{A}^n$. 
For sequences $f(n)$ and $g(n)$, $n=1,2,\dots$, the little-o notation $f(n)=o(g(n))$ means that
$\lim_{n\rightarrow\infty} f(n)/g(n)=0$, see~\cite[p.~61]{Landau1909}.

Random variables (RVs) are written with uppercase letters such as  $\rv{A}$ and
their realizations with lowercase letters $a$.
A probability mass function (pmf) of $\rv{A}$ is denoted by $\pmf{A}\in\set{P}$ or $\qmf{A}\in\set{P}$,
where $\set{P}$ is the set of pmfs for alphabet $\set{A}$.
We often discard subscripts on pmfs.
Pmfs or functions are also written as vectors, e.g., we write pmf $Q$ with
the $L$-letter alphabet $\set{A}=\{1,\ldots,L\}$ as $Q=[Q(1),\ldots,Q(L)]$.
The uniform pmf over a set of $K$ elements is denoted by $U_K$.

A random string $\rv{A}^n = \rv{A}_1\rv{A}_2\ldots\rv{A}_n$ has realizations $a^n = a_1a_2\ldots a_n \in \set{A}^n$.
The probability of a set $\set{S}\subseteq\set{A}^n$ of strings with respect to $P_{\rv{A}^n}$ is written as
\begin{equation}
    P_{\rv{A}^n}(\set{S}) = \sum_{a^n\in\set{S}} P_{\rv{A}^n}(a^n).
\end{equation}
We write $Q_{\rv{A}^n}=\qmf{A}^n$ for the pmf of a string of independent and
identically distributed (iid) RVs.

The $\ell_1$ distance between two pmfs $P$ and $Q$ on $\set{A}$ is
\begin{align}
    \done{P}{Q} = \sum_{a \in \set{A}} \left| P(a)-Q(a) \right|
\end{align}
and we have $\done{P}{Q}\le2$ with equality if and only if $\supp(P)\cap\supp(Q)=\emptyset$.
Of course, one may alternatively use the variational distance \cite[Ch.~11.6]{cover2006elements}.

The expectation of a real-valued function $f$ of a RV $\rv{A}$ with respect to the pmf $P$ is
denoted as
\begin{align}
  E_f(P) = \E{P}{f(\rv{A})} = \sum\limits_{a\in\supp(P)} P(a) f(a)
  \label{eq:expectation}
\end{align}
where $\supp(P)$ is the support of $P$, i.e., the set of $a\in\set{A}$ with $P(a)>0$.
The notation $E_f(P)$ is useful because we usually consider $P$ as the variable while $f$ is fixed.
Observe that $E_f(P)$ is linear in $P$.

The variance of $f(A)$ with respect to $P$ is $\Var{P}{f(A)}$ 
and the entropy of $P$ is $\entrp{P} = \E{P}{- \log_2 P(\rv{A})}$.
The cross entropy and divergence of pmfs $P$ and $Q$ with common alphabet $\set{A}$
are the respective
\begin{align}
  & \crossentrp{P}{Q} = \E{P}{- \log_2 Q(\rv{A})} \\
  & \diverg{P}{Q} = \E{P}{\log_2 \frac{P(\rv{A})}{Q(\rv{A})}}
\end{align}
and we have $\crossentrp{P}{Q}=\entrp{P}+\diverg{P}{Q} $.

\subsection{Empirical Probability}
Let $\set{P}_n\subseteq\set{P}$ be the set of pmfs that have rational values with
denominator $n$, also called the $n$-types. For a string $a^n$, let $n_i=n_i(a^n)$
be the number of occurrences of the letter $i$ for $i\in\set{A}$.
The empirical pmf, or type, of $a^n$ is
\begin{align}
   \pi(a^n)=\frac1n \left[ n_1,\dots,n_{L} \right]
   \label{eq:empirical}
\end{align}
and $\pi(a^n)\in\set{P}_n$. The number of $a^n$ with the same empirical pmf $\pi(a^n)$
is given by the multinomial
\begin{align}
    \binom{n}{n_1 \; \ldots \; n_{L}} = \frac{n!}{n_1! \dots n_L!}.
    \label{eq:type-class-size}
\end{align}

Let $\bar{n}_i = \sum_{a^n \in \set{A}^n} P_{\rv{A}^n}(a^n) \, n_i(a^n)$
be the average number of occurrences of letter $i$ for a specified $P_{\rv{A}^n}$.
The average empirical pmf is
\begin{align}
  \bar{\pi}
  = \sum_{a^n \in \set{A}^n} P_{\rv{A}^n}(a^n) \, \pi(a^n)
  = \frac1n [\bar{n}_1,\dots,\bar{n}_{L}] 
  \label{eq:Pavg}
\end{align}
and for any target pmf $\qmf{T}$ we have
\begin{align}
  \diverg{P_{\rv{A}^n}}{\qmf{T}^n}
  & = n \, \crossentrp{\bar{\pi}}{\qmf{T}} - \entrp{P_{\rv{A}^n}}.
  \label{eq:I-div-target}
\end{align}

\subsection{Four Lemmas}
\label{subsec:bounds-entropy-Idiv}

We state several lemmas; the first three are proved
in Appendixes~\ref{app:entropy-inequality2}-\ref{app:DM-Aary-increase}.
Let $p_{\rm min}$ and $p_{\rm max}$ be the respective minimum and maximum probabilities
of $P$, and define similar notation for the minima and maxima of the pmfs $Q$ and $R$.

\begin{lemma}
\label{lemma:entropy-inequality2}
Let $d_1=\done{P}{Q}$. If $d_1<2 p_{\rm min}$ then
\begin{align}
    \diverg{Q}{R} - \diverg{P}{R} \le \frac{d_1}{2} \log_2 \left( \frac{p_{\rm max}+d_1/2}{p_{\rm min}-d_1/2}
    \cdot \frac{r_{\rm max}}{r_{\rm min}} \right) . 
    \label{eq:entropy-inequality2}
\end{align}
\end{lemma}

Next, consider a threshold $I$ and define the pmf and string sets
$\set{E}=\{P: E_f(P)\le I\}$  and $\set{S}=\{a^n : \pi(a^n)\in\set{E}\}$.
We further define $P_{\rv{A}^n}(a^n)=Q^n(a^n)/Q^n(\set{S})$ for $a^n\in\set{S}$
and consider $\bar{\pi}$ to be a function of $I$.

\begin{lemma}
\label{lemma:d1-bound}
Consider any $P^*$ for which $E_f(P^*)=I$.  Then there is a $P\in\set{E}\cap\set{P}_n$ such that
\begin{align}
    & \done{P}{P^*} \le \frac{2 (L-1)}{n} . \label{eq:lemma-d1-bound}
\end{align}
\end{lemma}

\begin{lemma}
\label{lemma:DM-Aary-increase}
$E_f(\bar{\pi})$ is non-decreasing in $I$.
\end{lemma}

Next define
\begin{align}
   P^* = \argmin_{P\in\set{E}} \diverg{P}{Q}.
   \label{eq:Pstar}
\end{align}
We have the following standard result, see e.g. \cite[Ch.~11.5]{cover2006elements}.

\begin{lemma} \label{lemma:Pstar-explicit}
Let $f_{\rm min}=\min_{a\in\supp(Q)} f(a)$ and suppose that $f_{\rm min} < I < E_f(Q)$.
Then the solution of \eqref{eq:Pstar} is
\begin{align}
   & P^*(a) = \frac{Q(a) \, 2^{-\tau f(a)}}{\sum_{b\in\set{A}} Q(b) \, 2^{-\tau f(b)}}
   \label{eq:Pstar-explicit}
\end{align}
for $a\in\set{A}$, where $\tau$ is positive and chosen so that 
$E_f(P^*)=I$. Observe that $P^*(a)=0$ if and only if $Q(a)=0$.
\end{lemma}

\section{Bounds for Sample Means}
\label{sec:bounds-sums}

Consider the sample mean
\begin{align}
  \rv{S}_n=\frac1n\sum_{i=1}^n f(\rv{A}_i).
  \label{eq:f-sum}
\end{align} 
We are interested in characterizing $\Pr[\rv{S}_n \le I]$. 
Observe that $\rv{S}_n = E_f( \pi(\rv{A}^n) )$ so that
$\rv{S}_n \le I$ is the same as requiring $\pi(A^n)$ to lie in
the half-space $\set{E}=\{P: E_f(P) \le I\}$. 

\subsection{Deviations of the Sample Mean}
\label{subsec:brr}

We distinguish between two classes of RVs:
\begin{enumerate}
  \item lattice: there are constants $c_0$ and $d>0$ such that $\set{A} \subseteq \{c_0 + j d: j \text{ an integer}\}$;
  \item non-lattice.
\end{enumerate}
For example, every real-valued binary RV is lattice because one may choose $d$ as the absolute difference between the two alphabet values. In general, one chooses $d$ as the greatest common divisor\footnote{The greatest common divisor of a set of real numbers is the largest $d$ such that all all numbers in the set are integer multiples of $d$, e.g., the set $\{1,5/3\}$ has greatest common divisor $d=1/3$.} of the differences between consecutive possible values of $\set{A}$, see~\cite[Sec.~4]{bahadur1960-2}. 
The following large deviation theorem is developed in~\cite[Thm.~1]{bahadur1960-2}, see also \cite[Ch.~3.7]{Dembo-Zeitouni-98}.

\begin{theorem}[Bahadur and Ranga Rao] \label{theorem:brr}
Consider the iid string $\rv{A}^n$ with pmf $Q^n$ and the real-valued
function $f$ with domain $\set{A}$.
Let $f_{\text{min}}=\min_{a\in\supp(Q)} f(a)$ and consider an $I$ satisfying
\begin{align}
  f_{\text{min}} < I <  \mathbb{E}_Q[f(\rv{A})] = E_f(Q)
  \label{eq:I-bounds}
\end{align}
and also $\Pr[\rv{S}_n=I]>0$ if $f(\rv{A})$ is a lattice RV. We have
\begin{align}
    \Pr\left[ \rv{S}_n \le I \right] = \frac{2^{-n \diverg{P^*}{Q}}}{\sqrt{2 \pi n}} \, b \, (1+o(1))
    \label{eq:brr-div}
\end{align}
where $P^*$ is given by \eqref{eq:Pstar-explicit}. The values
\begin{align}
  & b = \left\{ \begin{array}{ll}
  \frac{1}{\alpha} \, \frac{\tau d}{1-e^{-\tau d}}, & \text{if $f(\rv{A})$ is a lattice RV} \\ 1/\alpha, & \text{else}
  \end{array} \right.
  \label{eq:bI} \\
  & \alpha = \sqrt{\Var{P^*}{ \ln \frac{P^*(\rv{A})}{Q(\rv{A})}}}
  \label{eq:alpha}
\end{align}
satisfy $0<b<\infty$ and $0<\alpha<\infty$, and $\tau$ is as in Lemma~\ref{lemma:Pstar-explicit}.
Moreover, given $I$ the $b$, $\alpha$, and $\tau$ are independent of $n$.
However, if $f(\rv{A})$ is lattice and $\Pr[\rv{S}_n=I]=0$ then one must replace
$b$ with $b_n=be^{-\tau d\, \theta_n}$ where $0\le\theta_n<1$. Thus, the right-hand side
of \eqref{eq:brr-div} is an upper bound on $\Pr\left[ \rv{S}_n \le I \right]$.
\end{theorem}

\begin{corollary}
\label{corollary:brr}
Choosing $Q=U_{L}$ in Theorem~\ref{theorem:brr}, \eqref{eq:brr-div} becomes
\begin{align}
    \Pr\left[ \rv{S}_n \le I \right] = \frac{L^{-n} \, 2^{n \entrp{P^*}}}{\sqrt{2 \pi n}} \,b\,(1+o(1))
    \label{eq:brr-bound}
\end{align}
where the value $b$ is as in~\eqref{eq:bI} and
\begin{align}
  \alpha = \sqrt{\Var{P^*}{\ln P^*(\rv{A})}}.
\end{align}
Furthermore, $P^*(a)>0$ for all $a\in\set{A}$.
\end{corollary}

\subsection{Example}
\label{subsec:multinomial-sums}

We apply Corollary~\ref{corollary:brr} to an example that corresponds to optimal DM~\cite{schulte2020invertible}.
Consider a target pmf $\qmf{T}$, e.g., a shaping pmf for energy-efficient communication.
Suppose we wish to count the number of $a^n$ with at least a specified probability with respect
to $\qmf{T}^n$, i.e., the number of $a^n$ satisfying
\begin{align}
    \qmf{T}^n(a^n) 
    \ge 2^{-n I}
    \label{eq:prob_condition}
\end{align}
and where equality holds for some $a^n$.
The bound \eqref{eq:prob_condition} is the same as
\begin{align}
    \frac{1}{n} \sum_{i=1}^n -\log_2 \qmf{T}(a_i) = \crossentrp{\pi(a^n)}{\qmf{T}} \le I
    \label{eq:prob_condition2}
\end{align}
which has the same form as \eqref{eq:f-sum} with $f(a)=-\log_2 \qmf{T}(a)$.
We thus restrict attention to the interval (see~\eqref{eq:I-bounds})
\begin{align}
   -\log_2 t_{\rm max} < I < E_f(Q) = \crossentrp{Q}{\qmf{T}}
  \label{eq:I-bounds2}
\end{align}
where $t_{\rm max}=\max_{a\in\set{A}}\qmf{T}(a)$. 

Now consider the pmf set $\set{E}=\{P:\crossentrp{P}{\qmf{T}}\le I\}$ and
the corresponding string set  $\set{S}=\{a^n : \pi(a^n)\in\set{E}\}$. We have
\begin{align}
    |\set{S}| = \sum_{P\in\set{E}\cap\set{P}_n} \binom{n}{nP(1)\ldots nP(L)}.
    \label{eq:multinomial_sum}
\end{align}
Normalizing \eqref{eq:multinomial_sum}  by $L^n$, one obtains a probability with respect to the uniform pmf over all strings in $\set{A}^n$.
We thus choose $\rv{A}^n$ with pmf $Q^n$ where $Q=U_L$ and can write
\begin{align}
    |\set{S}|
    = L^n \cdot \Pr\left[ \rv{S}_n \le I \right] .
    \label{eq:set-size-crossentrp}
\end{align}
Corollary~\ref{corollary:brr} gives the following result.

\begin{lemma} \label{theorem:brr-size}
The set $\set{S}$ of $a^n$ satisfying \eqref{eq:prob_condition2} has cardinality
\begin{align}
    |\set{S}| = \frac{2^{n \entrp{P^*}}}{\sqrt{2 \pi n}} \, b \, (1+o(1))
    \label{eq:brr-size}
\end{align}
where $P^*$ is given by \eqref{eq:Pstar-explicit}.
\end{lemma}

\begin{example} \label{example}
Fig.~\ref{fig:psimplex2} illustrates the probability simplex for $\set{A}=\{1,2,3\}$ and $n=10$.
Consider $Q=U_L$ and the lattice RV $f(\rv{A})$ with
\begin{align}
  f=-\log_2\qmf{T}=c_0+[0,1,8/3]
\end{align}
where $c_0\approx0.7290$, $\qmf{T}\approx[0.6033,0.3017,0.0950]$, and $d=1/3$. The bounds~\eqref{eq:I-bounds2} specify $c_0<I<c_0+11/9$. Consider $I\approx1.5290$ and $n=10$ so that $P_1=[7,0,3]/10$ and $P_2=[2,8,0]/10$ are both in $\set{P}_n$ and satisfy the constraint $E_f(P)=\crossentrp{P}{\qmf{T}}\le I$ with equality. We compute $P^*\approx[0.4915,0.3336,0.1749]$ and $\entrp{P^*}\approx 1.4720$. Consider the strings
\begin{align}
  & \hat{a}^n   = 1,1,1,1,1,1,2,2,3,3 \\
  & \tilde{a}^n = 1,1,1,1,1,2,2,2,2,3
\end{align}
with empirical pmfs $\pi(\hat{a}^n)=[6,2,2]/10$ and $\pi(\tilde{a}^n)=[5,4,1]/10$. The pmf $\pi(\hat{a}^n)$ maximizes the entropy in $\set{E}\cap\set{P}_n$ and $\pi(\tilde{a}^n)$ has the smallest $\ell_1$ distance to $P^*$ in $\set{E}\cap\set{P}_n$.
The entropies are $\entrp{\pi(\hat{a}^n)}\approx1.3710$ and $\entrp{\pi(\tilde{a}^n)}\approx1.3610$.
The $\ell_1$ distances are
$\done{P_1}{P^*}\approx0.6672$, $\done{P_2}{P^*}\approx0.9328$,
$\done{\pi(\hat{a}^n)}{P^*}\approx0.2672$, and $\done{\pi(\tilde{a}^n)}{P^*}\approx0.1498$.
\end{example}

\begin{figure}[t]
    \centering

\pgfdeclarelayer{background}
\pgfdeclarelayer{foreground}
\pgfsetlayers{background,main,foreground}

\begin{tikzpicture}
\begin{ternaryaxis}[
	xlabel=$P(1)$,
	ylabel=$P(2)$,
	zlabel=$P(3)$,
	minor tick num=1, 
	grid=both,
	ternary limits relative=false,
	xmin=0,
    xmax=1,
    ymin=0,
    ymax=1,
    zmin=0,
    zmax=1,
    clip=false,
]
    \begin{pgfonlayer}{background}
    	\addplot3[fill=red!10] coordinates {
            (0.2, 0.8, 0.0)
            (0.7, 0.0, 0.3)
            (1.0, 0.0, 0.0)
    	};
	\end{pgfonlayer}{background}
	\draw[-] (0.2, 0.8, 0.0) -- (0.7, 0.0, 0.3); 
	
	\node [pin={[pin edge={black,thick,-}, black]-90:uniform},thick,draw=black] at (0.3333,0.3333) {};
	\node [pin={[pin edge={black,thick,-}, black]80:$\qmf{T}$},thick,draw=red] at (0.6033,0.3017) {};
	\node [pin={[pin edge={black,thick,-}, black]0:$P_1$},thick,draw=black] at (0.7,0) {};
	\node [pin={[pin edge={black,thick,-}, black]-10:$P_2$},thick,draw=black] at (0.2,0.8) {};
	\node [pin={[pin edge={black,thick,-}, black]-10:$P^*$},thick,draw=red] at (0.4917,0.3333) {};
	\node [pin={[pin edge={black,thick,-}, black]-30:$\pi(\hat{a}^n)$},thick,draw=blue] at (0.6,0.2) {};
	\node [pin={[pin edge={black,thick,-}, black]-90:$\quad \pi(\tilde{a}^n)$},thick,draw=blue] at (0.5,0.4) {};
\end{ternaryaxis}
\end{tikzpicture}
    \caption{Probability simplex for $\set{A}=\{1,2,3\}$ with grid lines for $n=10$.
    The red shaded area shows the set $\set{E}$ of pmfs $P$ satisfying $\crossentrp{P}{\qmf{T}}\le I$ where  $\qmf{T}\approx[0.6033,0.3017,0.0950]$ and $I\approx1.5290$.
    }
    \label{fig:psimplex2}
\end{figure}
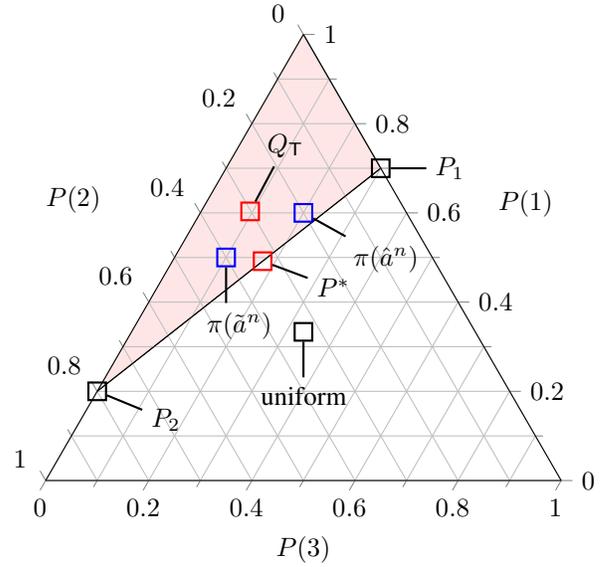

\subsection{Convergence of the Average Empirical Distribution}
\label{subsec:average-epmf-scaling}
We develop a scaling result for the conditional limit theorem~\cite[Ch.~11.6]{cover2006elements}.
Let $p_{\rm min}^*$ and $p_{\rm max}^*$ be the respective minimum and maximum probabilities of $P^*$.
Suppose that $q_{\rm min}>0$ so that $p_{\rm min}^*>0$ by Lemma~\ref{lemma:Pstar-explicit}.
Define the function
\begin{align}
  \Delta(P) & = \diverg{P}{Q}-\diverg{P^*}{Q}. \label{eq:DP}
\end{align}
Observe that $\Delta(P)$ is bounded and convex in $P$,
and non-negative for $P\in\set{E}$. We have the following scaling result.

\begin{theorem} \label{theorem:scaling-bounds}
Consider the iid string $\rv{A}^n$ with pmf $Q^n$.
If $I$ satisfies \eqref{eq:I-bounds} then we have
\begin{align}
  \diverg{\bar{\pi}}{P^*} \le \Delta(\bar{\pi}) & < \frac{1}{n} \, 2^{\tilde{c}+1}\,(1+o(1))
  \label{eq:diverg1-DMbound}
\end{align}
where
\begin{align}
   \tilde{c} = (L-1) \log_2 \frac{ p_{\rm max}^* \, q_{\rm max} }{ p_{\rm min}^* \, q_{\rm min} } .
   \label{eq:limc}
\end{align}
\end{theorem}
\begin{proof}
Observe that $E_f(\bar{\pi})\le I$ by the linearity of expectation.
The first inequality in \eqref{eq:diverg1-DMbound} thus follows by a
``Pythagorean'' inequality~\cite[Ch.~11.6]{cover2006elements}.
The second inequality is proved in Appendix~\ref{app:scaling-bounds}.
\end{proof}

\section{Divergence for Distribution Matchers}
\label{sec:DM-scaling}

The problem of DM, as described in~\cite{schulte2020invertible,schulte2017divergence}, is to
choose a codebook $\set{S}$ of strings $a^n$ so that the divergence
$\diverg{P_{\rv{A}^n}}{\qmf{T}^n}$ is minimized where $P_{\rv{A}^n}=U_{|\set{S}|}$.
Without loss of essential generality, we assume $\supp(\qmf{T})=\set{A}$
and $\qmf{T}\ne U_{L}$. We have $\entrp{P_{\rv{A}^n}}=\log_2 |\set{S}|$, 
and the rate in bits per symbol is $\Rinfo=\frac{1}{n} \log_2 |\set{S}|$.

The paper~\cite[Prop.~3]{schulte2020invertible} shows that optimal $\set{S}$
have all strings $a^n$ satisfying $\crossentrp{\pi(a^n)}{\qmf{T}}\le I$ for some $I$, see~\eqref{eq:prob_condition2}.
The paper~\cite{schulte2017divergence} shows that for $L=2$ the divergence $\diverg{P_{\rv{A}^n}}{\qmf{T}^n}$ of
the best $\set{S}$ scales as $\frac{1}{2}\log_2 n$ with $n$, i.e., binary
DMs cannot provide low divergence.
We prove the same result for general discrete alphabets $\set{A}$ by using
Lemmas \ref{lemma:entropy-inequality2}-\ref{theorem:brr-size} and Theorem \ref{theorem:scaling-bounds}.

\subsection{Small $I$}
\label{sec:DM-small-I}

The case $I\le-\log_2 t_{\rm max}$ makes sense only if equality holds.
$\set{S}$ is thus the set of $a^n$ having only letters with maximal probability $\qmf{T}(a)=t_{\rm max}$.
The positive $P_{\rv{A}^n}(a^n)$ then all have the same value $1/N_{\rm max}^{n}$ where
$N_{\rm max}$ is the number of letters $a$ with probability $\qmf{T}(a)=q_{\rm max}$.
Moreover, we compute 
\begin{align}
  & \diverg{P_{\rv{A}^n}}{\qmf{T}^n} = n \log_2 \frac{1}{N_{\rm max}\,t_{\rm max}} 
\end{align}
and the divergence grows linearly with $n$ since $N_{\rm max}\,t_{\rm max}<1$ by our assumption that $\qmf{T}$ is not uniform.

\subsection{Large $I$}
\label{sec:DM-large-I}

For $I\ge\crossentrp{U_L}{\qmf{T}}$, the following result is the analog of Lemma~17 in \cite[App.~B]{schulte2020invertible}.

\begin{lemma}
\label{lemma:DM-Aary-range}
$\diverg{P_{\rv{A}^n}}{\qmf{T}^n}$ scales linearly with $n$ for $I>\log_2 L$.
\end{lemma}
\begin{proof}
We have $\log_2 L < \crossentrp{U_{L}}{\qmf{T}}$ because $\qmf{T}\ne U_L$.
Moreover, if $\log_2L<I<\crossentrp{U_{L}}{\qmf{T}}$ then
Theorem~\ref{theorem:scaling-bounds} gives $\bar{\pi}\rightarrow P^*$ as $n\rightarrow\infty$
and thus $\crossentrp{\bar{\pi}}{\qmf{T}}\rightarrow\crossentrp{P^*}{\qmf{T}}=I$
since $\crossentrp{P^*}{\qmf{T}}=I$ by Lemma~\ref{lemma:Pstar-explicit}.
But $\entrp{P_{\rv{A}^n}}\le n\log_2L$ and thus the divergence \eqref{eq:I-div-target}
grows linearly with $n$. Finally, Lemma~\ref{lemma:DM-Aary-increase} states
that $\crossentrp{\bar{\pi}}{\qmf{T}}$ is non-decreasing in $I$.
\end{proof}

\subsection{Intermediate $I$}
\label{sec:DM-intermediate-I}

For $-\log_2 t_{\rm max}<I<\crossentrp{U_L}{\qmf{T}}$, Lemma~\ref{theorem:brr-size} gives
\begin{align}
    \entrp{P_{\rv{A}^n}} = -\frac{1}{2}\log_2 n + n \entrp{P^*} + \log_2 \frac{b}{\sqrt{2\pi}} + o(1) .
    \label{eq:logS-bound}
\end{align}
Inserting~\eqref{eq:logS-bound} into~\eqref{eq:I-div-target} gives
\begin{align}
  \diverg{P_{\rv{A}^n}}{\qmf{T}^n}
  & = \frac12 \log_2 n - n\left[ \entrp{P^*} - \entrp{\bar{\pi}} \right] \nonumber \\
  & \quad +n\diverg{\bar{\pi}}{\qmf{T}} - \log_2 \frac{b}{\sqrt{2\pi}} + o(1)
  \label{eq:div-expand3}
\end{align}
which is the analog of~\cite[Eq.~(107)]{schulte2020invertible}.
Theorem~\ref{theorem:scaling-bounds} with $Q=U_L$ now gives
\begin{align}
  0 \le \Delta(\bar{\pi}) = \entrp{P^*} - \entrp{\bar{\pi}} < \frac{1}{n} \, 2^{\tilde{c}+1}\,(1+o(1)).
  \label{eq:entropy-expand}
\end{align}
Inserting \eqref{eq:entropy-expand} into \eqref{eq:div-expand3} and taking $n\rightarrow\infty$, we have
\begin{align}
  & - 2^{\tilde{c}+1} - \log_2 \frac{b}{\sqrt{2\pi}} \nonumber \\
  & \quad \le \liminf_{n\to\infty} \left( \diverg{P_{\rv{A}^n}}{\qmf{T}^n} - \frac12 \log_2 n - n\diverg{\bar{\pi}}{\qmf{T}} \right) \nonumber\\
  & \quad \le - \log_2 \frac{b}{\sqrt{2\pi}}.
  \label{eq:diverg-liminf}
\end{align}
The divergence $\diverg{P_{\rv{A}^n}}{\qmf{T}^n}$ thus grows at least as $\frac12 \log_2 n$ with $n$. 
The next result shows that one can achieve this growth.

\begin{theorem}\label{theorem:DM-divergence}
Finite-alphabet block DMs that minimize the divergence $\diverg{P_{\rv{A}^n}}{\qmf{T}^n}$ have
$\diverg{P_{\rv{A}^n}}{\qmf{T}^n}$ that grows as $\frac12 \log_2 n$ with $n$
and $\Rinfo\rightarrow\entrp{\qmf{T}}$ as $n\rightarrow\infty$.
\end{theorem}
\begin{proof}
We have already proved the converse. For the coding theorem,
one may quantize $\qmf{T}$ to a $\tilde{Q}_{\rv{T}}\in\set{P}_n$ satisfying
$\done{\tilde{Q}_{\rv{T}}}{\qmf{T}}\le L/(2n)$~\cite[Prop.~2]{boecherer-geiger-IT16}.
Now choose $I=\entrp{\tilde{Q}_{\rv{T}}}$ so that $P^*=\tilde{Q}_{\rv{T}}$.
Theorem~\ref{theorem:scaling-bounds} gives
$n\diverg{\bar{\pi}}{\tilde{Q}_{\rv{T}}}< 2^{\tilde{c}+1}(1+o(1))$, and using \eqref{eq:div-expand3}
and the left-hand side of \eqref{eq:entropy-expand} we find that the divergence $\diverg{P_{\rv{A}^n}}{\qmf{T}^n}$ grows as
$\frac12 \log_2 n$. Furthermore, $\Rinfo\rightarrow\entrp{\qmf{T}}$ by 
Lemmas~\ref{lemma:entropy-inequality2} and~\ref{theorem:brr-size}.
\end{proof}

%

\section*{Acknowledgement}
The author wishes to thank P.~Schulte and D.~Lentner for discussions, and the reviewers for comments that improved the presentation.
This work was supported by the German Research Foundation (DFG) under Grant {KR 3517/9-1}.

\bibliographystyle{IEEEtran}
\bibliography{IEEEabrv,conf-jnls,references}

\setcounter{section}{0}
\renewcommand{\thesection}{\Alph{section}}
\renewcommand{\thesubsection}{\arabic{subsection}}
\renewcommand{\appendix}[1]{%
  \refstepcounter{section}%
  \par\begin{center}%
    \begin{sc}%
      Appendix \thesection\par\nobreak%
      #1%
    \end{sc}%
  \end{center}\nobreak%
}

\appendix{Proof of Lemma~\ref{lemma:entropy-inequality2}}
\label{app:entropy-inequality2}
The condition $p_{\rm min}>d_1/2$ implies $q_{\rm min}>0$ and therefore
$\supp(P)=\supp(Q)=\set{A}$. Consider the pmfs $P$ and
$Q=P+\Delta$ where $\sum_a \Delta(a)=0$ and $\sum_a |\Delta(a)|=d_1$. Define
\begin{align}
  \Delta_+ = \sum_{a:\,\Delta(a)>0} \Delta(a), \quad
  \Delta_- = \sum_{a:\,\Delta(a)<0} \Delta(a)
\end{align}
and observe that $\Delta_{+}=-\Delta_{-}=d_1/2$. Now expand
\begin{align}
  & \diverg{Q}{R} - \diverg{P}{R} \nonumber \\
  & = - \diverg{P}{Q} + \sum_{a: \Delta(a)\ne0} \Delta(a) \log_2 \frac{P(a) + \Delta(a)}{R(a)}
  \label{eq:app-ei2-1}
\end{align}
and compute the bounds
\begin{align}
  & \sum_{a: \Delta(a)>0} \Delta(a) \underbrace{\log_2 \frac{P(a) + \Delta(a)}{R(a)}}_{\le \log_2 \frac{p_{\rm max} + d_1/2}{r_{\rm min}}}
  \le \frac{d_1}{2} \log_2 \frac{p_{\rm max} + d_1/2}{r_{\rm min}} \label{eq:app-ei2-2} \\
  & \sum_{a: \Delta(a)<0} \Delta(a) \underbrace{\log_2 \frac{P(a) + \Delta(a)}{R(a)}}_{_{\ge \log_2 \frac{p_{\rm min} - d_1/2}{r_{\rm max}}}}
  \le -\frac{d_1}{2} \log_2 \frac{p_{\rm min} - d_1/2}{r_{\rm max}}. \label{eq:app-ei2-3}
\end{align}

\appendix{Proof of Lemma~\ref{lemma:d1-bound}}
\label{app:d1-bound}
The proof applies the approach in~\cite[Sec.~III and Sec.~VI.C]{boecherer-geiger-IT16} with small changes. 
Order the values $P^*(i)$, $i=1,\dots,L$, such that $f(i)\le f(j)$ if $i<j$. Define the pmf $P$ via
\begin{align}
  P(i) = \left\{
  \begin{array}{ll}
    \lfloor n P^*(i) \rfloor/n, & i=2,\dots,L \\
    1 - \sum_{i=2}^{L} P(i), & i=1
  \end{array} \right.
\end{align}
and note that $P\in\set{P}_n$. Define $e(i)=P^*(i)-P(i)$ for all $i$.
We have $0\le e(i)<\frac1n$ for $i=2,\dots,L$ and $\sum_{i=1}^{L} e(i)=0$
so that $-(L-1)/n<e(1)\le0$. We thus have
  \eqref{eq:lemma-d1-bound} and
\begin{align}
    E_f(P) = \underbrace{E_f(P^*)}_{=I} - \sum_{i=1}^{L} e(i) \underbrace{f(i)}_{\ge f(1)} \le I
    \label{eq:lemma-d1bound-crossentrp}
\end{align}
where the last step uses $e(i)\ge0$ for $i=2,\dots,n$. We thus have $P\in\set{E}\cap\set{P}_n$.

\appendix{Proof of Lemma~\ref{lemma:DM-Aary-increase}}
\label{app:DM-Aary-increase}
Suppose we increase $I$ from $I_1$ to $I_2$ with
the corresponding string sets $\set{S}_1$ and $\set{S}_2$
and average empirical pmfs $\bar{\pi}_1$ and $\bar{\pi}_2$, respectively. 
Suppose $I_2-I_1$ is sufficiently large so that $\set{S}_2\setminus\set{S}_1\ne\emptyset$.
The definition \eqref{eq:Pavg} gives
\begin{align}
  E_f(\bar{\pi}_2)
  & = \sum_{a^n\in\set{S}_2} \frac{Q^n(a^n)}{Q^n(\set{S}_2)} E_f(\pi(a^n)) \nonumber \\
  & = \frac{Q^n(\set{S}_1)}{Q^n(\set{S}_2)} E_f(\bar{\pi}_1)
  + \sum_{a^n\in\set{S}_2\setminus\set{S}_1} \frac{Q^n(a^n)}{Q^n(\set{S}_2)} E_f(\pi(a^n)) \nonumber \\
  & > E_f(\bar{\pi}_1)
\end{align}
where the last step follows by $E_f(\pi(a^n)) > I_1 \ge E_f(\bar{\pi}_1)$
for $a^n\in\set{S}_2\setminus\set{S}_1$. 

\appendix{Proof of Theorem~\ref{theorem:scaling-bounds}}
\label{app:scaling-bounds}

There may be no $P\in\set{E}\cap\set{P}_n$ such that $E_f(P)=I$.
However, for each $n$ there is an $I^*$ such that $E_f(P) \le I^* \le I$ for all
$P\in\set{E}\cap\set{P}_n$, and with $E_f(P)=I^*$ for some $P\in\set{E}\cap\set{P}_n$. 
Define the optimized pmf (note that $I^*$ replaces $I$)
\begin{align}
    P^* = & \underset{P}{\arg \min} \; \diverg{P}{Q} \nonumber \\
    & \text{subject to } E_f(P) = I^*
    \label{eq:Pstar-app}.
\end{align}

Next, by convexity of $\Delta(P)$ we have
\begin{align}
  \Delta(\bar{\pi}) & \le \sum_{a^n \in \set{S}} P_{\rv{A}^n}(a^n) \, \Delta(\pi(a^n)).
  \label{eq:diverg1-bound}
\end{align}
Now partition $\set{S}$ into the sets $\{\set{S}_j\}_{j=0}^N$
where $\set{S}_j$ is the set of $a^n\in\set{S}$ satisfying
\begin{align}
  j\,\delta \le \Delta(\pi(a^n)) < (j+1)\,\delta
  \label{eq:app-diverg1-bounds}
\end{align}
for small positive $\delta$ and sufficiently large $N$. Let $\Delta_{\max}$ be an upper bound on $\Delta(P)$.
We select a $N_1<N$ and use \eqref{eq:app-diverg1-bounds} and $\Delta(P) \le \Delta_{\max} < \infty$ to write
\begin{align}
  & \Delta(\bar{\pi}) < \left( \sum_{j=0}^{N_1-1}
    P_{\rv{A}^n}(\set{S}_j) (j+1) \right) \delta + P_{\rv{A}^n}\left( \cup_{j=N_1}^N \set{S}_j \right) \Delta_{\max} \nonumber \\
  & = \left( \sum_{j=0}^{N_1-1}
    P_{\rv{A}^n} \left( \cup_{i=j}^{N_1-1} \set{S}_i \right) \right) \delta + P_{\rv{A}^n}\left( \cup_{j=N_1}^N \set{S}_j \right) \Delta_{\max} 
  \label{eq:diverg1-bound2}
\end{align}
where the second step follows because $\{\set{S}_j\}_{j=0}^N$ is a partition.

Let $\set{T}_j=\cup_{i=j}^{N} \set{S}_i$ and note that $\set{T}_j\ne\emptyset$
for small $\delta$ and $j$. Define the following parameters of $\set{T}_j$:
\begin{align}
  & a_j^n = \underset{a^n\in\set{T}_j}{\arg \max} \; E_f(\pi(a^n)), \quad I_j^* = E_f(\pi(a_j^n)), \\
  & P_j^* = \underset{P}{\arg \min} \; \diverg{P}{Q} \text{ subject to } E_f(P) = I_j^*
\end{align}
and observe that $I_0^*=I^*$ and $P_0^*=P^*$. Observe also that $a^n\in\set{T}_j$
if and only if $E_f(\pi(a^n))\le I_j^*$ because $\diverg{P_j^*}{Q}$ increases as $I_j^*$ decreases.
In fact, since $\tau$ is the Lagrange multiplier for the optimization problem in \eqref{eq:Pstar}, we have
\begin{align}
   \frac{d\,\diverg{P^*}{Q}}{d\, I} = - \tau < 0.
\end{align}

Let $p_{{\rm min},j}^*$ and $p_{{\rm max},j}^*$ be the respective minimum and maximum probabilities of $P_j^*$. For $j\,\delta$ not too large we have $p_{{\rm min},j}^*>0$, see Lemma~\ref{lemma:Pstar-explicit}. We further define
\begin{align}
  \hat{a}_j^n = \underset{a^n\in\set{S}_j}{\arg \min} \; \diverg{\pi(a^n)}{Q}
\end{align}
and note that $\hat{a}_j^n \ne a_j^n$ in general, see Fig.~\ref{fig:psimplex2}.
Next, if $n > (L-1)/ p_{{\rm min},j}^*$ then Lemmas~\ref{lemma:entropy-inequality2} and~\ref{lemma:d1-bound}
imply that there is an $\tilde{a}_j^n\in\set{S}_j$ such that
\begin{align}
  & \done{\pi(\tilde{a}_j^n)}{P_j^*} \le \frac{2 (L-1)}{n} \\
  & \diverg{\pi(\tilde{a}_j^n)}{Q} - \diverg{P_j^*}{Q} \le \frac{c_j}{n}
  \label{eq:diverg-bound2}
\end{align}
where
\begin{align}
  c_j = (L-1) \log_2 \left( \frac{p_{\rm max,j}^* + \frac{L-1}{n}}{p_{\rm min,j}^* - \frac{L-1}{n}}
  \cdot \frac{q_{\rm max}}{q_{\rm min}} \right).
  \label{eq:Deltaj}
\end{align}
Note that $\tilde{a}_j^n \ne \hat{a}_j^n$ in general, see Fig.~\ref{fig:psimplex2}.

By definition we have $\diverg{\pi(\hat{a}_j^n)}{Q} \le \diverg{\pi(\tilde{a}_j^n)}{Q}$ so that
\eqref{eq:app-diverg1-bounds} and \eqref{eq:diverg-bound2} give
\begin{align}
  & j \delta \le \Delta(\pi(\hat{a}_j^n)) \le \Delta(\pi(\tilde{a}_j^n)) \le \Delta(P_j^*) + \frac{c_j}{n} .
  \label{eq:app-diverg-diff-bound}
\end{align}
We can thus use Theorem~\ref{theorem:brr} to write
\begin{align}
   P_{\rv{A}^n} \left( \set{T}_j \right)
   & = \frac{\Pr\left( E_f(\pi(A^n)) \le I_j^* \right)}{\Pr\left(  E_f(\pi(A^n)) \le I_0^* \right)} \nonumber \\
   & \le 2^{-n \Delta(P_j^*)} \frac{b(I_j^*)}{b(I_0^*)} \, (1+o(1)) \label{eq:Sj-diverg-bound}
\end{align}
where the $b(I_j^*)$ are the $b$s in Theorem~\ref{theorem:brr} corresponding to $P_j^*$.
The inequality in \eqref{eq:Sj-diverg-bound} accounts for lattice RVs
for which one may have $\Pr[S_n=I_j^*]=0$ for $j=1,\dots,N_1$, see Theorem~\ref{theorem:brr}.

We now choose $\delta=1/n$ and $N_1=\lceil \sqrt{n}\, \rceil$ so that the expressions
\eqref{eq:diverg1-bound2}, \eqref{eq:app-diverg-diff-bound}, and \eqref{eq:Sj-diverg-bound} give
\begin{align}
  \Delta(\bar{\pi}) < \left[ \left(\sum_{j=0}^{\infty} 2^{-j} \right) \frac{1}{n} + \Delta_{\max}\, 2^{-\sqrt{n}} \right] b_r\,2^{c}\,(1+o(1))
  \label{eq:diverg1-bound3}
\end{align}
where
\begin{align}
    b_r = \max_{0 \le j \le N_1} \frac{b(I_j^*)}{b(I_0^*)}, \quad c = \max_{0 \le j \le N_1} c_j.
    \label{eq:brc}
\end{align}
Next, since $N_1/n\rightarrow0$ we have $(j+1)\delta \rightarrow 0$ in \eqref{eq:app-diverg1-bounds} and  hence $P_j^*\rightarrow P^*$ for all $j=0,1,\dots,N_1$. We thus have $p_{\rm min,j}^* \rightarrow p_{\rm min}^*$ and $p_{\rm max,j}^* \rightarrow p_{\rm max}^*$. Moreover, by continuity of $\alpha$ and $\tau$ with respect to the threshold $I$, we have $b(I_j^*)\rightarrow b$ for all $j=0,1,\dots,N_1$ and therefore $b_r\rightarrow1$.
Finally, evaluating the sum in \eqref{eq:diverg1-bound3} gives \eqref{eq:diverg1-DMbound}.

\end{document}